\documentclass[preprint,12pt]{elsarticle}

\usepackage{amssymb}
 \usepackage{amsthm}

\journal{Mathematical Biosciences}

\def\E{{\text{E}}}

\usepackage{subfig}
\usepackage{bm}

\usepackage{amsmath}

\newtheorem{theorem}{Theorem}[section]
\newtheorem{proposition}{Proposition}[section]
\newtheorem{lemma}[theorem]{Lemma}
\newtheorem{corollary}{Corollary}[theorem]

\usepackage{color}

\begin{document}

\begin{frontmatter}



\title{Impact of heterogeneity on infection probability: Insights from single-hit dose-response models}


\author[inst1]{Francisco J. P\'erez-Reche}
\ead{fperez-reche@abdn.ac.uk}
\affiliation[inst1]{organization={School For Natural and Computing Sciences, SUPA},
            addressline={University of Aberdeen}, 
            city={Aberdeen},
            postcode={AB24 3UE}, 
            country={United Kingdom}}

%

\begin{abstract}
The process of infection of a host is complex, influenced by factors such as microbial variation within and between hosts as well as differences in dose across hosts. This study uses dose-response and within-host microbial infection models to delve into the impact of these factors on infection probability. It is rigorously demonstrated that within-host heterogeneity in microbial infectivity enhances the probability of infection. The effect of infectivity and dose variation between hosts is studied in terms of the expected value of the probability of infection. General analytical findings, derived under the assumption of small infectivity, reveal that both types of heterogeneity reduce the expected infection probability. Interestingly, this trend appears consistent across specific dose-response models, suggesting a limited role for the small infectivity condition. Additionally, the vital dynamics behind heterogeneous infectivity are investigated with a within-host microbial growth model which enhances the biological significance of single-hit dose-response models. Testing these mathematical predictions inspire new and challenging laboratory experiments that could deepen our understanding of infections. 
\end{abstract}




\end{frontmatter}


\section{Introduction}

A host exposed to a pathogenic microbe can become infected. However, infection is a complex process that depends on numerous factors, including the number of microbes to which the host has been exposed (i.e. the dose), the virulence and genetic variability of the pathogens, as well as the susceptibility of the host to such pathogens~\cite{Milgroom_2023_BookInfections,Gog_Epidemics2015_SevenChallengesWithinHost,
Mideo_Alizon_TREE2008_WithinBetweenHost,
Didelot_NatRevMicrobio2016_WithinHostEvolBacterialPathogens,
Restif_PhilTransRSB2015_WithinHostIntroduction,Metcalf-Grenfell_PLOSPathogens2020_WithinHostSARSCoV-2,Lythgoe-Hall_Science2024_WithinHostDiversitySARSCoV2}. Given the typical uncertainty about all these factors, it is useful to describe the occurrence of infection as a random event that can occur with certain probability. 

Estimating the probability of infection following exposure to a pathogenic microbe is a crucial step in quantitative microbial risk assessment~\cite{Haas_Book2014}. Heterogeneity associated with the host-microbe interaction within hosts and differences between hosts in terms of the ingested dose and susceptibility, are expected to influence the chances of infection. However, a systematic understanding of the role of heterogeneity on infection is lacking. Interpreting infection as a biological invasion could guide intuition on the role of heterogeneity in the probability of infection. Different types of heterogeneity have been considered by ecological and epidemiological models for the probability of biological invasions~\cite{Melbourne_EcolLett2007_HeterogeneousInvasion,
Su_EcolRes2009_LandscapeHeterogeneityHostParasite,
Miller_PRE2007,PerezReche_JRSInterface2010,
Neri_PLoSCBio2011}. Such models, however, suggest that depending on the type of heterogeneity, the probability of infection could either be enhanced of reduced. 

Here, a general and mathematically rigorous analysis of the role of several types of heterogeneities on the probability of infection is presented within the context of dose-response models. These are mathematical functions that map a dose to the probability of infection~\cite{Haas_Book2014}.  Dose-response models have been effectively utilised in previous works to characterise the probability of infection across diverse microbial species and host populations \cite{Gifford_JTheorBiol1969_DoseResponse,Furumoto1967a,Furumoto1967,
Haas_AmJEpidemiology1983,Teunis_RiskAnal2000,
Strachan_2005_DoseResponseO157,Chen_RiskAnal2006,teunis_hierarchical_2008,Teunis-Strachan_2010_DoseResponseSalmonella,Watanabe_RiskAnalysis2010_DoseResponseSARSCoV2,
Conlan_JRoySocInterface2011,Hamilton_WaterResearch2017_DoseResponse,Xu_JTheorBiol2023_WithinHostDoseResponseSARSCoV2}. 

Among several possible mechanistic approaches proposed to formulate dose-response relationshipts, the single-hit framework is the most widely used  \cite{Haas_AmJEpidemiology1983,Teunis_RiskAnal2000,Haas_Book2014}. It is based on the following two assumptions:
\begin{description}
 \item[Assumption 1.] Ingestion of a single pathogenic microbe is sufficient to infect a host. 
 \item[Assumption 2.] If a host ingests a dose containing multiple pathogenic microbes, the event of one microbe causing infection is independent of the event of another microbe causing infection. 
\end{description}

A central quantity of single-hit models is the infectivity, which denotes the probability that any microbe within an ingested dose will infect the host. The necessity to account for heterogeneity in infectivity and dose has long been recognised by the scientific community employing dose-response models~\cite{Moran_JHygiene1954_DoseResponse,Furumoto1967a,Furumoto1967,
Gifford_JTheorBiol1969_DoseResponse}. Three sources of heterogeneity can be naturally considered: 
\begin{description}
\item[(I) Within-host heterogeneous infectivity.] If a host ingests several pathogenic microbes, the infectivity may be different for different microbes. 
\item[(II) Between-host heterogeneous infectivity.] Two identical microbes may have different infectivity in different hosts. 
\item[(III) Between-host heterogeneous dose.] The ingested dose may be different for different hosts.     
\end{description}

The systematic investigation of whether each type of heterogeneity enhances or diminishes the probability of infection has not been systematically undertaken in previous studies. Furthermore, a link between heterogeneous infectivity and microbial infection mechanisms (e.g., microbial rates of birth and death) has yet to be established \cite{Haas_EnvSciTech2014_DoseResponseReview}. Advancing in these two directions constitutes the central motivation of the present work.

To address the first point, two possible definitions of the probability of infection will be considered to address the first point: The probability that a given host becomes infected after ingesting a dose with specific infectivities and the average probability of infection for a set of hosts. Previous works have predominantly focused on the latter definition, as it is well-suited for interpreting experimental data.  In such cases, variability has been addressed by treating infectivity and/or dose as random variables. This has led to dose-response formulas that map the expected value of the dose to the expected value of the probability of infection, as derived from infectivity and dose probability distributions. Despite the widespread use of this modeling approach, a comprehensive analysis of how the probability of infection varies with infectivity dispersion remains lacking. 

For instance, the exact beta-Poisson model or its widely-used approximation, the ``approximate'' beta-Poisson formula, assume Poisson and beta distributions for the dose and infectivity, respectively~\cite{Furumoto1967,Teunis_RiskAnal2000}. The shape parameters of the infectivity distribution,  usually denoted as $\alpha$ and $\beta$, encode the dependence of the probability of infection on the variance of infectivity. However, the link between these parameters and the variance of infectivity has not been exploited to understand the effect of such variance on the expected probability of infection. The approximate beta-Poisson formula has also been used to mathematically justify an ubiquitous flattening of the dose-response curve that cannot be explained assuming homogeneous infectivity~\cite{Haas_Book2014}. However, an explicit link between the slope and the variance of infectivity was again not established. Indeed, the demonstration was based on the so-called median effective dose and parameter $\alpha$ in such a way that both the mean and variance of the infectivity vary simultaneously. A complete proof of whether the decrease of the slope is directly related to the infectivity variance is still lacking.

Concerning the second direction mentioned above, the growth dynamics of microbes within a host have been described with birth-death models that have been related to dose-response models~\cite{Armitage_Nature1965_BirthDeathMicrobialInfection,
Williams_JRStatSoc1965_BirthDeathMicrobialInfections}. However, heterogeneous infectivity was not incorporated into these models~\cite{Haas_EnvSciTech2014_DoseResponseReview}, and for instance, a mechanistic interpretation of the parameters $\alpha$ and $\beta$ of the beta-Poisson model has not been provided.

Overall, the current understanding of how variability in dose and/or infectivity affects the probability of infection is limited, even for relatively simple models like the approximate beta-Poisson formula. This paper addresses this gap by presenting general results applicable to any single-hit dose-response model with heterogeneous dose and/or infectivity. These findings are then demonstrated using specific dose-response models.

The remaining of the paper is structured as follows: Section~\ref{sec:General-SingleHit} introduces a general single-hit response model and provides general results for the probability of a host becoming infected by a dose of microbes with heterogeneous infectivity (i.e., heterogeneity of type I).  A particular case of this model with homogeneous dose and infectivity is presented in Sec.~\ref{sec:modelA}. Subsequently, Section~\ref{sec:ModelsExpectedPinf} gradually extends the homogeneous model to analyse the effect of all types of heterogeneity on the expected value of the probability of infection (see a schematic representation in Fig.~\ref{fig_Models_space_3D}). Well-known dose-response models like the exponential and beta-Poisson formulas \cite{Haas_Book2014,Haas_EnvSciTech2014_DoseResponseReview} are special cases of such models. Sec.~\ref{sec:GrowthModel} presents a population growth model for microbial infection, considering heterogeneous birth and death rates. Sec.~\ref{sec:conclusion} finalises the paper with a summary of the main findings, followed by a discussion on the scientific significance of experimentally validating the mathematical predictions in laboratory settings, along with the potential challenges involved.


\begin{figure}
\centering
\includegraphics[width=11 cm]{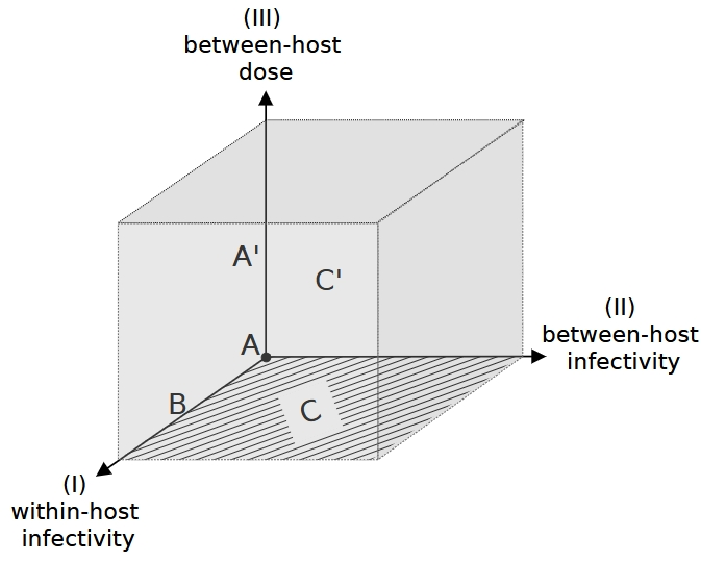}
\caption{Single-hit models for the expected probability of infection. Each axis represents a different source of heterogeneity: (I) Randomly distributed infectivity within hosts, (II) randomly distributed infectivity for different hosts and (III) randomly distributed dose for different hosts. Model $A$ is located at the origin where both infectivity and dose are homogeneous. Models of type $A^\prime$ are located along the axis corresponding to heterogeneous doses. Models of type $B$ are located along the axis corresponding to random infectivity within hosts and homogeneous between-hosts infectivity and dose. Models of type $C$ are located on the horizontal plane with homogeneous dose but heterogeneous infectivity within and between hosts. Models of type $C^\prime$ incorporate all three types of heterogeneity and are located in the bulk of the model space.}
\label{fig_Models_space_3D}
\end{figure}






\section{General single-hit dose-response model}
\label{sec:General-SingleHit}
Let us assume a population of $H$ hosts such that each host ingest (or is in contact with) $n_h$ pathogenic microbes. The infectivity of a microbe $i$ is quantified by the probability $x_{i,h}$ that the microbe infects the host. From assumption 1 in the introduction, any microbe with $x_{i,h}>0$ can cause infection. Following assumption 2, the probability that a host is infected is
\begin{equation}
\label{eq:Ph_general}
    P_h\left(\{x_{i,h}\}_{i=1}^{n_h},n_h\right)=1-\prod_{i=1}^{n_h} (1-x_{i,h})~.
\end{equation}
Assuming that the infection of any host is independent of the rest of the hosts, the number of infected hosts obeys a Poisson-Binomial distribution, $I\sim \text{PoissonBin}(H,\{P_h\}_{h=1}^H)$ \cite{Tang-Tang_StatScience2023_PoissonBinomial}. The mean and variance of the number of infected individuals are then given by 

\begin{equation}
\label{eq:E-Var_general}
    \text{E}(I) = \sum_{h=1}^H P_h~,\\
    \text{Var}(I) = \sum_{h=1}^H P_h(1-P_h)~.
\end{equation}

Within this model, type I heterogeneity is captured by the dependence of $x_{i,h}$ on the microbe $i$. Heterogeneity of type II is encoded by the dependence of $x_{i,h}$ on the host $h$. Heterogeneity of type III is accounted for by the dependence of $n_h$ on the host.

After deriving a useful formula for the probability of infection of a host in Lemma~\ref{lemma:Formula_Ph}, a general result is proven in Theorem~\ref{th:Within-Host-Inf} for the effect of within-host heterogeneous infectivity on the probability of infection.

\begin{lemma}
\label{lemma:Formula_Ph}
    The probability of infection of a host that ingests a dose of $n$ microbes with infectivities $\{x_{i}\}_{i=1}^n$ can be expressed as 
    \begin{equation}
    \label{eq:Ph_general_theorem}
        P_h\left(\{x_{i}\}_{i=1}^{n},n\right) = 1-(1-\bar{x})^n \exp{\left(-n \sum_{k=1}^\infty \frac{M_k}{k (1-\bar{x})^k} \right)}~,
    \end{equation}
    in terms of the sample mean,  $\bar{x}=n^{-1} \sum_{i=1}^n x_i$,  and central moments, $M_k = n^{-1} \sum_{i=1}^n (x_i-\bar{x})^k$, of the infectivity. Here, the notation has been simplified by dropping the explicit dependence of infectivities and dose on the host $h$.
\end{lemma}

\begin{proof}
    To derive Eq.~\eqref{eq:Ph_general_theorem}, it is convenient to express Eq.~\eqref{eq:Ph_general} as
    \begin{equation}
    \label{eq:pro:Formula_Ph_proof1}
        P_h\left(\{x_{i}\}_{i=1}^{n},n\right) =1-\left[\exp{\left(\frac{1}{n} \sum_{i=1}^n \ln (1-x_i) \right)}\right]^n~.
    \end{equation}

The identity
\begin{equation}
\ln(1-x_i) = \ln(1-\bar{x})+\ln\left(1-\frac{x_i-\bar{x}}{1-\bar{x}} \right)~
\end{equation}
allows Eq.~\eqref{eq:pro:Formula_Ph_proof1} to be written as
\begin{equation}
    \label{eq:pro:Formula_Ph_proof2}
        P_h\left(\{x_{i}\}_{i=1}^{n},n\right) =1-(1-\bar{x})^n \left[\exp{\left(\frac{1}{n} \sum_{i=1}^n 
        \ln\left(1-\frac{x_i-\bar{x}}{1-\bar{x}} \right)
         \right)}\right]^n~.
    \end{equation}

The argument of the exponential function in Eq.~\eqref{eq:pro:Formula_Ph_proof2} can be expressed as
\begin{equation}
\label{eq:pro:Formula_Ph_proof3}
\frac{1}{n} \sum_{i=1}^n  \ln \left(1-\frac{x_i-\bar{x}}{1-\bar{x}} \right) = - \sum_{k=1}^\infty \frac{1}{n} \sum_{i=1}^n \frac{\left( \frac{x_i-\bar{x}}{1-\bar{x}}\right)^k}{k} = -\sum_{k=1}^{\infty} \frac{M_k}{k(1-\bar{x})^k}~.
\end{equation}
Here, the first equality follows from the series expansion $\ln(1-z_i) = -\sum_{k=1}^\infty z_i^k$, where $z_i=(x_i-\bar{x})/(1-\bar{x})$.

Introducing Eq.~\eqref{eq:pro:Formula_Ph_proof3} into Eq.~\eqref{eq:pro:Formula_Ph_proof2}, leads to Eq.~\eqref{eq:Ph_general_theorem}.
\end{proof}

\begin{theorem}
\label{th:Within-Host-Inf} 
 Assume that a host ingests a dose of $n$ microbes with infectivities $\{x_{i,h}\}_{i=1}^n$ and let $\bar{x} = n^{-1} \sum_i^n x_{i,h}$ be the mean infectivity.  The probability of infection of such host is minimal when the infectivity is homogeneous, i.e., when all microbes have the same infectivity, $x_{i,h} = \bar{x}$.
\end{theorem}

\begin{proof}
The function $\ln(1-z_i)$ is non-convex for any $z_i \leq 1$ so that, by Jensen's inequality \cite{Jensen_ActaMathematica1906_Inequality}, 
\begin{equation}
\label{th:Within-Host-Inf_proof2}
\frac{1}{n} \sum_{i=1}^n \ln(1-z_i) \leq \ln \left(1- \frac{1}{n} \sum_{i=1}^n z_i \right)~,
\end{equation}
and this implies
\begin{equation}
\label{th:Within-Host-Inf_proof3}
\frac{1}{n} \sum_{i=1}^n  \ln \left(1-\frac{x_i-\bar{x}}{1-\bar{x}} \right) \leq 0,
\end{equation}
if $z_i=(x_i-\bar{x})/(1-\bar{x})$. From Eqs.~\eqref{th:Within-Host-Inf_proof3} and \eqref{eq:pro:Formula_Ph_proof3},
\begin{equation}
\label{th:Within-Host-Inf_proof4}
-\sum_{k=1}^{\infty} \frac{M_k}{k(1-\bar{x})^k} \leq 0~,
\end{equation}
and, from Eq.~\eqref{eq:Ph_general_theorem}, one obtains
\begin{equation}
\label{th:Within-Host-Inf_proof5}
P_h\left(\{x_{i}\}_{i=1}^{n},n\right) \geq 1-(1-\bar{x})^n~.
\end{equation}
The theorem follows from Eq.~\eqref{th:Within-Host-Inf_proof5}. Indeed, the equality only holds for situations with homogeneous within-host infectivity such that $M_k=0$ for every $k\geq 2$. The inequality holds strictly for any infection with heterogenous within-host infectivity with mean $\bar{x}$.
\end{proof}

The general model given by Eq.~\eqref{eq:Ph_general} is conceptually interesting. However, practical applications of dose-response models often deal with limited data consisting in a dose-response curve, i.e., the probability of infection estimated for a set of different doses. These data do not allow the parameters of the general model to be properly inferred (i.e. the infectivities for all microbes ingested by a host, $\{x_{i,h}\}_{i=1}^{n_h}$). It is more practically feasible to assume that the different sources of heterogeneity can be described in terms of a smaller number of parameters for the probability distribution of infectivity and dose. This motivates the adoption of models for the expected value of the probability of infection from a pragmatic standpoint. Such models are also crucial for comprehending the general effect of heterogeneities of types II and III, and will be utilised for this purpose in Sec.~\ref{sec:ModelsExpectedPinf}.

\section{Model $A$: Homogeneous dose and infectivity}
\label{sec:modelA}

This model assumes that all hosts ingest the same number of microbes and that all microbes have the same infectivity within all hosts, i.e., $n_h=n$ and $x_{i,h}=x$ for any $h$ and $i$.

Under these assumptions, Eq.~\eqref{eq:Ph_general} reduces to
\begin{equation}
\label{eq:P_ModelA}
    P^{(A)}(x,n)= 1-(1-x)^n~,
\end{equation}
and the number of infected hosts obeys a binomial distribution, $I \sim \text{Bin}(n,P^{(A)})$.

Consistent with intuition, this model suggests that the probability of infection increases with both infectivity and dose.

\section{Models for the expected value of the probability of infection}
\label{sec:ModelsExpectedPinf}

This section begins by presenting results for the expected value of the probability of infection, applicable in scenarios with small infectivities across any dose distribution. Subsequently, these findings will be demonstrated through specific models $A^{\prime}-C^{\prime}$ introduced in subsections \ref{subsec:ModelAprime}-\ref{subsec:ModelCprime}.

\subsection{General results for the expected probability of infection}
\label{subsec:GeneralExpectedPinf}

Theorem \ref{th:GeneralExpectedPinf} gives a general result for the dependence of the expected probability of infection on the infectivity and/or dose variability across different hosts.

\begin{theorem}
\label{th:GeneralExpectedPinf}
Let us suppose that the following conditions hold:
\begin{description}
   \item{(a)} The probability of infection for a \emph{given} infectivity and dose is given by $P^{(A)}(x,n)$ (Eq.~\ref{eq:P_ModelA}). 
   \item{(b)} The between-host variation of infectivity is described by a continuous random variable defined on $[0,1]$ with expectation $\mu_x$ and variance $v_x$.
   \item{(c)} The dose ingested by different hosts is a discrete random variable taking values in $\mathbb{Z}^+$ (i.e., at least one microbe is ingested), with expectation $\mu_n$ and variance $v_n$.
\end{description}

Then, in the limit of small infectivity, the expected probability of infection, $\hat{P}$, is a decreasing function of both $v_x$ and $v_n$.
\end{theorem}

\begin{proof}
    The proof essentially uses a cumulant expansion procedure \cite{Riley-Hobson-Bence_MathsBook_3rdEd}. From assumption (a), the expected probability of infection is given by
    \begin{equation}
    \label{eq:Phat}
        \hat{P} = 1 - \E_x ( \E_n ((1-x)^n))~,
    \end{equation}
where $\E_x$ and $\E_n$ are the expectations corresponding to infectivity and dose, respectively.

The binomial power in Eq.~\eqref{eq:Phat} can be expanded to give $\E_x ( \E_n ((1-x)^n))=1+y$, where
\begin{equation}
\label{eq:En_expansion_1}
        y\equiv-\left(\mu_x + \frac{1}{2} (v_x+\mu_x^2)+O(x^3)\right)\mu_n+\frac{1}{2} (v_x+\mu_x^2+O(x^3))(v_n+\mu_n^2)+O(x^3n^3)~.
\end{equation}
Here, the identities $\E_x(x^2)=v_x+\mu_x^2$ and $\E_n(n^2)=v_n+\mu_n^2$ have been used to express $y$ as a function of expected values and variances of $x$ and $n$.

It is now convenient to express the expected value in Eq.~\eqref{eq:Phat} as $\E_x ( \E_n ((1-x)^n))=\exp(\ln(1+y))$ and use the expansion $\ln{(1+y)}=y-y^2/2+O(y^3)$. After introducing the expression \eqref{eq:En_expansion_1} into this expansion and neglecting $O(x^3)$ terms in the resulting function of $x$, one obtains the following approximation:
\begin{equation}
\label{eq:hatPApprox}
    \hat{P} \simeq f_1(\mu_x,v_x,\mu_n,v_n) \equiv 1- \exp \left(-\mu_x \mu_n + \frac{\mu_x^2}{2} (v_n-\mu_n)+\frac{v_x}{2}(\mu_n^2-\mu_n+v_n) \right)~.
\end{equation}

According to Eq.~\eqref{eq:hatPApprox}, $\hat{P}$ decreases with increasing $v_n$ for any $\mu_x>0$. $\hat{P}$ also decreases with increasing $v_x$ since $\mu_n^2-\mu_n+v_n\geq 0$ for any reasonable dose with $\mu_n \geq 1$ and $v_n \geq 0$.
\end{proof}

Corollary \ref{cor:n50_general} establishes a general link between the slope of the dose-response curve and the variance of the infectivity. This gives mathematical support to the ubiquitous flattening of the dose-response curve experimentally observed and modelled by assuming heterogeneous infectivity \cite{Furumoto1967,Haas_Book2014}.

\begin{corollary}
\label{cor:n50_general}
In the limit of small infectivity, the slope of $\hat{P}$ evaluated at the median dose---defined as the dose $\mu_n=\mu_{50}$ for which the expected probability of infection is 50\%---is an increasing function of $\mu_x$ and a decreasing function of $v_x$ for any dose distribution.
\end{corollary}

\begin{proof}
    The median dose, $\mu_{50}$, corresponds to $\hat{P}=1/2$ which, in the limit of small infectivity, reduces to $f_1(\mu_x,v_x,\mu_{50},v_n)=1/2$ (cf. Eq.~\ref{eq:hatPApprox}). From this condition, the rate of variation of $\hat{P}$ at $\mu_n=\mu_{50}$ can be expressed as 
    \begin{equation}
    \label{eq:cor_n50_general}
        \left.\frac{\text{d} \hat{P}}{\text{d}\mu_n} \right|_{\mu_n=\mu_{50}} = \frac{\mu_x}{2} \left(1+\frac{\mu_x}{2} \right)-\frac{v_x}{2}\left(\mu_{50}-\frac{1}{2} \right)~.   
    \end{equation}
    This equation trivially proves the corollary for any reasonable median dose $\mu_{50}>1/2$.
\end{proof}

\subsection{Models of type $A^{\prime}$: Homogeneous infectivity and heterogenous dose}
\label{subsec:ModelAprime}
Models of type $A^{\prime}$ extend Model A to include heterogeneity of type III (see Figure \ref{fig_Models_space_3D}). Specifically, it is posited that the doses ingested by different hosts are independent and identically distributed (\emph{iid}) random variables drawn from a probability mass function $p_n^{\text{d}}(\bm\xi)$, where $\bm\xi$ are the distribution parameters. The expected probability of infection can then be expressed in terms of $P^{(A)}(x,n)$ as follows:
\begin{equation}
\label{eq:PAprime}
    P^{(A^{\prime})}(x,\bm\xi) = \sum_{n =1}^{\infty} P^{(A)}(x,n) p^{\text{d}}(n;\bm\xi)~.
\end{equation}

The subsequent part of this section illustrates the effect of heterogeneous doses on $P^{(A^{\prime})}(x,\bm\xi) $ for two commonly assumed scenarios for $p^{\text{d}}(n;\bm\xi)$.

\subsubsection{Example $A^{\prime}_1$: Poisson-distributed dose}
\label{subsec:A1prime}
Assuming that the number of particles ingested by a host is a Poisson distribution with mean $\mu_n$ (i.e. $n_h \sim \text{Pois}(\mu_n)$) leads to the well-known exponential dose-response model for the probability of infection \cite{Haas_AmJEpidemiology1983,Haas_Book2014}:
\begin{equation}
\label{eq:PAprime1}
    P^{(A^{\prime}_1)}(x,\mu_n) = 1-e^{-x\mu_n}~.
\end{equation}
This expression follows from Eq.~\eqref{eq:PAprime} by using the probability mass function, $p^{\text{d}}(n;\mu_n) = \mu_n^n e^{-\mu_n} / n!$, for the Poisson distribution.


\subsubsection{Example $A^{\prime}_2$: Negative binomial dose distribution}
\label{subsec:A2prime}
The Poisson distribution assumes $\mu_n=v_n$ and does not allow the effect of $v_n$ on the expected probability of infection to be independently analysed from $\mu_n$. In fact, doses can be characterised by a larger dispersion than that allowed by a Poisson distribution (i.e. $v_n>\mu_n$)~\cite{Anderson-May_1991Book,Haas_RiskAnalysis2002,teunis_hierarchical_2008}. Over-dispersion can be modelled by assuming a negative binomial distribution, $n_h \sim \text{NB}(r,p)$, with parameters  $r>0$ and $p\in[0,1]$. The corresponding probability function is
\begin{equation}
    p^{\text{d}}(n_h;r,p)=\frac{\Gamma(n_h+r)}{n_h! \Gamma(r)} (1-p)^{n_h} p^r~,
\end{equation}
which, after introducing into \eqref{eq:PAprime1}, results in the following expected probability of infection: 
\begin{equation}
\label{eq:PAprime2_r_p}
    P^{(A^{\prime}_2)}(x,r,p) = 1- \left(\frac{p}{p+x(1-p)} \right)^r~.
\end{equation}

The parameters $r$ and $p$ are related to the mean and variance of the dose, $\mu_n$ and $v_n$, as follows:
\begin{equation}
\label{eq:Aprime2_Relation_rp_muv}
    r=\frac{\mu_n^2}{v_n-\mu_n} \text{ and } p=\frac{\mu_n}{v_n}~.
\end{equation}

Using these expressions in Eq.~\eqref{eq:PAprime2_r_p} allows the expected probability of infection to be written as a function $P^{(A^{\prime}_2)}(x,\mu_n,v_n)$ of the mean and variance of the dose.

Note that the condition $r>0$ implies $v_n \geq \mu_n$ for the negative binomial distribution. The smallest possible dispersion for a negative binomial distribution corresponds to $v_n = \mu_n$ which is the condition defining a Poisson distribution. The following proposition shows that the expected probability of infection in examples $A^{\prime}_1$ and $A^{\prime}_2$ are identical when $v_n = \mu_n$.

\begin{proposition}
\label{pro:Aprime2_reduces_to_Aprime1}
    Model $A^{\prime}_2$ reduces to model $A^{\prime}_1$ when the variance of the dose approaches the mean from above (i.e. in the limit $v_n \searrow \mu_n$).
\end{proposition}

\begin{proof}
    To prove the proposition, it is convenient to express the expected probability of infection in terms of $r$ and $\mu_n$. Combining Eqs.~\eqref{eq:PAprime2_r_p} and \eqref{eq:Aprime2_Relation_rp_muv}, leads to 
    \begin{equation}
    \label{eq:PAprime2_mu_r}
        P^{(A^{\prime}_2)}(x,\mu_n,r)=1-\left(\frac{r}{r+x\mu_n} \right)^r~.
    \end{equation}

    From Eq.~\eqref{eq:Aprime2_Relation_rp_muv}, the limit $v_n \searrow \mu_n$ corresponds to $r \rightarrow \infty$. Through simple algebraic manipulations, one obtains the probability \eqref{eq:PAprime1} for $r \rightarrow \infty$:
    \begin{equation}
        \lim_{r \rightarrow \infty} P^{(A^{\prime}_2)}(x,\mu_n,r) = 1-\lim_{r \rightarrow \infty} \left( 1+\frac{-x \mu_n}{-r}\right)^{-r}=1-e^{-x\mu_n}~.
    \end{equation}
\end{proof}

The decrease of the infection probability when increasing the variance $v_n$, holds for $P^{(A^{\prime}_2)}(x,r,p)$ even without the small infectivity requirement of Theorem \ref{th:GeneralExpectedPinf}. This is shown in Proposition \ref{pro:PAprime2_variance} after proving the following lemma:

\begin{lemma}
\label{le:fyr}
The function 
    \begin{equation}
        f(y,r) = y+(r+y) \ln \left( \frac{r}{r+y}\right)~
    \end{equation}
satisfies $f(y,r) \leq 0$ for any $y \geq 0$ and $r>0$.
\end{lemma}
\begin{proof}
This can be easily demonstrated by plotting $f(y,r)$ as a function of $y$ for different values of $r$. Alternatively, a rigorous proof for any $r>0$ can be given from the following two facts: 
\begin{description}
    \item[(i)] $f(y,r)$ has only one zero and is located at $y=0$ for any $r>0$. To show this, it is convenient to define a new variable $z=\frac{r+y}{er}$ such that $f(y,r)=0$ reduces to $z \ln z =-e^{-1}$. The solution to this equation is $z = e^{-1}$ which corresponds to $y=0$.
    \item[(ii)] $f(y,r)$ is a decreasing function of $y$ in the interval $y>0$. Indeed, $\partial f(y,r)/\partial y = - \ln \left(\frac{r+y}{r} \right) <0$. 
\end{description}

Following these results, $f(y,r)$ must be non-positive for any $y>0$ and $r>0$.
\end{proof}

\begin{proposition}
\label{pro:PAprime2_variance}
    For a fixed mean dose, the expected probability of infection of model $A^{\prime}_2$ is a decreasing function of the dose variance.
\end{proposition}

\begin{proof}
    It is again convenient to use expression  \eqref{eq:PAprime2_mu_r} for $P^{(A^{\prime}_2)}(x;\mu_n,r)$. 

    The variation of $P^{(A^{\prime}_2)}$ with $v_n$ at fixed $\mu_n$ can be calculated as
    \begin{equation}
    \label{eq:partial_P_v}
        \frac{\partial P^{(A^{\prime}_2)}}{\partial v_n} = \frac{\partial P^{(A^{\prime}_2)}}{\partial r} \frac{\partial r}{\partial v_n}~.
    \end{equation}

    On the one hand, Lemma \ref{le:fyr} ensures that 
    \begin{equation}
    \label{eq:partial_P_r}
        \frac{\partial P^{(A^{\prime}_2)}}{\partial r}=-\frac{r^r}{(r+x\mu_n)^{r+1}} f(x \mu_n,r) \geq 0~
    \end{equation}
    for any $x\mu_n \geq 0$ and $r>0$. On the other hand, the following inequality holds:
    \begin{equation}
    \label{eq:partial_r_v}
        \frac{\partial r}{\partial v_n} = -\frac{\mu_n^2}{(v_n-\mu_n)^2}<0~.
    \end{equation}

    From Eq.~\eqref{eq:partial_P_v}, the inequalities \eqref{eq:partial_P_r} and \eqref{eq:partial_r_v} imply that $\partial P^{(A^{\prime}_2)}/\partial v_n <0$ and this proves the proposition.
\end{proof}

The next proposition relates all the models of type $A$ and $A^\prime$ and gives further evidence that increasing dose variability decreases the expected probability of infection for any given infectivity.

\begin{proposition}
\label{pro:Inequalities_PA}
    The following inequalities hold for any infectivity $x$ and dose distribution parameters $\mu_n$ and $v_n>\mu_n$:
    \begin{equation}
        P^{(A)}(x,\mu_n) \geq P^{(A^{\prime}_1)}(x,\mu_n) \geq P^{(A^{\prime}_2)}(x,\mu_n,v_n)~.
    \end{equation}

\end{proposition}

\begin{proof}
    The first inequality trivially follows from the comparison of Eqs.~\eqref{eq:P_ModelA} and \eqref{eq:PAprime1} and the fact that $e^{-x\mu_n} \geq (1-x)^{\mu_n}$ for any $\mu_n>0$. The second inequality follows from Propositions \ref{pro:Aprime2_reduces_to_Aprime1} and \ref{pro:PAprime2_variance}. Fig.~\ref{fig_P_ModelsA_Ap1_Ap2} demonstrates the inequalities for particular values of $\mu_n$ and $v_n$.
\end{proof}

\begin{figure}
\centering
\includegraphics[width=8 cm]{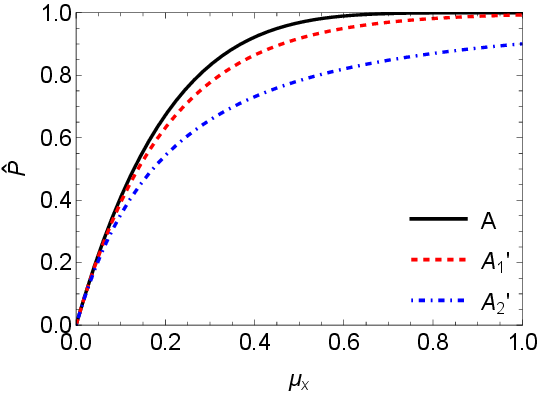}
\caption{Graphical illustration of the inequalities given in Proposition \ref{pro:Inequalities_PA}. The dependence of the probability of infection on the infectivity $x$ for a system with homogeneous infectivity (Model $A$, continuous black line) is compared to the dependence on $x$ of the expected probability of infection of models $A^{\prime}_1$ (dashed red line) and $A^{\prime}_2$ (dot-dashed blue line). The dose was set to $n=5$ for model $A$. The mean and variance of the dose were set to $\mu_n=5$ and $v_n=20$, respectively, for models $A^{\prime}_1$ and $A^{\prime}_2$.}
\label{fig_P_ModelsA_Ap1_Ap2}
\end{figure} 

\subsection{Models of type B: Heterogeneous infectivity within hosts and homogeneous dose}
\label{subsec:ModelB}

Models of type B extend Model A to account for heterogeneity of type I, by assuming that the infectivity depends on the microbe type but not on the host (see Figure \ref{fig_Models_space_3D}). Mathematically, this corresponds to $x_{i,h} = x_i$ within any host, $h$. The infectivities $\{x_i\}_{i=1}^n$ of the $n$ microbes in a dose are described as \emph{iid} random variables drawn from a distribution with probability density function (PDF) $\rho_{\text{m}}(x;\bm\xi_{\text{B}})$. Here, $\bm\xi_{\text{B}}$ are the parameters of the distribution. The expected value of the probability of infection under these conditions is~\cite{Haas_RiskAnalysis2002}:
\begin{equation}
\label{eq:Pinv_B}
    P^{(B)}(\mu_m(\bm\xi_{\text{B}}),n)=\prod_{i=1}^n \int_0^1 P^{(A)}(x_i,n) \rho_{\text{m}}(x_i;\bm\xi_{\text{B}}) \text{d} x_i= 1-(1-\mu_m(\bm\xi_{\text{B}}))^n~,
\end{equation}
where $\mu_m(\bm\xi_{\text{B}})=\int_0^1 x \rho_{\text{m}}(x;\bm\xi_{\text{B}}) \text{d}x$ is the expected value of the infectivity within a host. 

Interestingly, the expected probability $P^{(B)}$ depends solely on the infectivity parameter $\mu_m$; the variability in $x$ does not influence $P^{(B)}$. In other words, when a set of hosts ingest doses with a mean infectivity of $\mu_m$, the mean probability of infection calculated across different hosts is solely determined by $\mu_m$. This contrasts with the probability that \emph{an individual host} becomes infected, which, for a given mean infectivity $\bar{x}$, increases with greater heterogeneity in infectivity (Theorem~\ref{th:Within-Host-Inf}).

The relationship between $P^{(B)}$ and $\mu_m$ mirrors the relationship between $P^{(A)}$ and $x$. Consequently, for models of type $B$, it can be inferred that the expected probability of infection for a dose of $n$ microbes with \emph{iid} random infectivities can effectively be represented by a set of $n$ microbes with homogeneous infectivity $\mu_m$. Following this premise, the infectivity within hosts in models of type $C$ and $C^\prime$ will be denoted as $x$, despite the possibility for such models to account for variability in infectivity within hosts by simply substituting $x$ with $\mu_m$.

\subsection{Models of type $C$: Heterogeneous infectivity between hosts and homogeneous dose}
\label{subsec:ModelC}

Models of type $C$ represent an expansion of models of type $B$, introducing a host-dependent mean infectivity to accommodate heterogeneities of type I and II (see Figure \ref{fig_Models_space_3D}).  The expansion of model $A$ to include heterogeneity of type II represents a specific case wherein all microbes within a host share the same infectivity, denoted as $x_{i,h}=x_h$ for every microbe $i$, yielding $\mu_m=x_h$. Within a set of $H$ hosts, the mean infectivities, $\{x_h\}_{h=1}^H$, are regarded as \emph{iid} random variables drawn from a distribution with probability density function (PDF) $\rho_h(x;\bm\xi_{\text{C}})$. The probability of infection for a randomly selected host is expressed as follows:
\begin{equation}
\label{eq:Pinv_C}
    P^{C}(\bm\xi_{\text{C}},n) = \int_0^1 P^{(A)}(x,n) \rho_h(x;\bm\xi_{\text{C}}) \text{d}x~.
\end{equation}

The general results presented in Sec.~\ref{subsec:GeneralExpectedPinf} for small infectivity hold for models of type $C$ which are particular cases of the model in Sec.~\ref{subsec:GeneralExpectedPinf} with $v_n=0$.

An alternative family of models of type $C$ could be defined by assuming that the parameters of the distribution of infectivity within hosts, $\bm\xi_{\text{B}}$, vary across different hosts instead of assuming a variation of $x_h$. A formulation along these lines will be used in Sec.~\ref{sec:GrowthModel} to establish a link between single-hit dose-response models and microbial growth processes. From a practical viewpoint, however, assuming that the infectivity is a random variable as in Eq.~\eqref{eq:Pinv_C} has proven useful for many applications of dose-response models and this is the option chosen to illustrate models of type $C$ in the following example.

\subsubsection{Example C$_1$: Beta-distributed between-host infectivity}
\label{sec:C1}
It is common to assume that the randomness of infectivity between hosts can be described with a beta distribution~\cite{Moran_JHygiene1954_DoseResponse,Furumoto1967}, i.e. $x \sim \text{Beta}(\alpha,\beta)$, where $\alpha$ and $\beta$ are positive parameters. The corresponding PDF is
\begin{equation}
\label{eq:rho_beta}
    \rho_h^{(C_1)}(x;\alpha,\beta)=\frac{\Gamma(\alpha+\beta)}{\Gamma(\alpha) \Gamma(\beta)}x^{\alpha-1}(1-x)^{\beta-1}~.
\end{equation}
The beta distribution has support in the interval $x\in[0,1]$ and is therefore a mathematically reasonable description of the infectivity which is ultimately a probability. 

Introducing the PDF from Eq.~\eqref{eq:rho_beta} into Eq.~\eqref{eq:Pinv_C}, yields an expected probability of infection
\begin{equation}
\label{eq:Pinf_C1}
    P^{(C_1)}(\alpha,\beta,n)=1-\frac{\Gamma(\alpha+\beta) \Gamma(n+\beta)}{\Gamma(\beta) \Gamma(n+\alpha+\beta)}~.
\end{equation}

Expressing $\rho_h^{(C_1)}(x;\alpha,\beta)$ and $P^{(C_1)}(\alpha,\beta,n)$ in terms of $\alpha$ and $\beta$ is typical in dose-response studies assuming a beta-distributed infectivity. To understand the impact of heterogeneity, however, it is useful to express $\alpha$ and $\beta$ as a function of the expected value, $\mu_x$, and variance, $v_x$, of the infectivity as follows:
\begin{equation}
\label{eq:alpha_beta_vs_mux_vx}
    \alpha = \mu_x \left(\frac{\mu_x (1-\mu_x)}{v_x}-1 \right)~, 
    \quad
    \beta = (1-\mu_x) \left(\frac{\mu_x (1-\mu_x)}{v_x}-1 \right)~.
\end{equation}
Note that the conditions $\alpha, \beta >0$ impose an upper bound for the variance:
\begin{equation}
\label{eq:vmax}
v_x < v_{\max}(\mu_x) \equiv \mu_x (1-\mu_x)~.
\end{equation}
In particular, this implies an absolute upper bound for the variance: $v_x<1/4$. 

Four regimes can be defined in terms of the dependence of $\rho_h^{(C_1)}(x;\alpha,\beta)$ on $x$ (see Fig.~\ref{fig_rhox_4regimes}). The regimes depend on whether each of the parameters $\alpha$ and $\beta$ is larger or smaller than one:

\begin{description}
 \item[Regime (i) - Highly heterogeneous infectivity:] This regime corresponds to $\alpha,\beta < 1$. The infectivity PDF $\rho_h^{(C_1)}$ diverges for $x=0$ and $x=1$ (see Fig.~\ref{fig_rhox_4regimes}(a)). This describes populations of hosts that tend to have a bimodal distribution of infectivities: The microbe can be highly infective for some hosts whereas it tends to be harmless for others. The relative proportion of hosts at risk of infection is larger when $\beta<\alpha$ as this makes the divergence of $\rho_h^{(C_1)}$ at $x=1$ dominant over the one at $x=0$. 
 
 In this regime, $\mu_x \in (0,1)$ and the variance takes relatively high values, $v_x \in (\max \{b_1(\mu_x),b_2(\mu_x)\},v_{\max}(\mu_x))$ (see Fig.~\ref{fig_4Regimes_mux_vx}). Here,
\begin{align}
    b_1(\mu_x)\equiv \frac{\mu_x^2 (1-\mu_x)}{1+\mu_x}~,
    b_2(\mu_x)\equiv \frac{\mu_x (1-\mu_x)^2}{2-\mu_x}~,
\end{align}
are boundary functions for $v_x$ corresponding to the conditions $\alpha=1$ and $\beta=1$ (derived from Eqs.~\eqref{eq:alpha_beta_vs_mux_vx}). It is important to note that regime (i) is described as a high heterogeneity regime since it allows the largest possible values of the variance $v_{\max}$ for given  $\mu_x$. However, this does not imply that $v_x$ will necessarily take large values since, for instance,  $v_{\max}(\mu_x)$ tends to zero for $\mu_x \rightarrow 0$ or $\mu_x \rightarrow 1$.
 
 \item[Regime (ii) - High infectivity with intermediate heterogeneity:] $\alpha>1$ and $\beta < 1$. The PDF $\rho_h^{(C_1)}$ is zero at $x=0$ and monotonically increases with increasing $x$ to exhibit a divergence at $x=1$ (see Fig.~\ref{fig_rhox_4regimes}(b)). The mean infectivity takes high values, $\mu_x \in (1/2,1)$, and the variance takes intermediate values, $v_x \in (b_2(\mu_x),b_1(\mu_x))$ (Fig.~\ref{fig_4Regimes_mux_vx}).
 \item[Regime (iii) - Low infectivity with intermediate heterogeneity:] $\alpha<1$ and $\beta > 1$. This regime is complementary to regime (ii) in the sense that $\rho_h^{(C_1)}$ diverges at $x=0$ and monotonically decreases with increasing $x$ until $\rho_h^{(C_1)}=0$ is reached at $x=1$ (see Fig.~\ref{fig_rhox_4regimes}(c)). This regime corresponds to situations with relatively small mean infectivity, $\mu_x \in (0,1/2)$, and intermediate values of the variance, $v_x \in (b_1(\mu_x),b_2(\mu_x))$.
 \item[Regime (iv) - Low heterogeneity:] $\alpha > 1$ and $\beta > 1$. In this regime, $\rho_h^{(C_1)}$ is zero for $x=0,1$ and takes positive values for any $x\in(0,1)$ (see Fig.~\ref{fig_rhox_4regimes}(d)). The mean infectivity can take any admissible value, $\mu_x \in (0,1)$, and the variance takes low values, $v_x \in (0, \min \{b_1(\mu_x),b_2(\mu_x)\})$. 
\end{description}

In all regimes, the proportion of hosts at high risk of infection increases for decreasing $\beta$ and increasing $\alpha$. 

\begin{figure}
\centering
\includegraphics[width=12 cm]{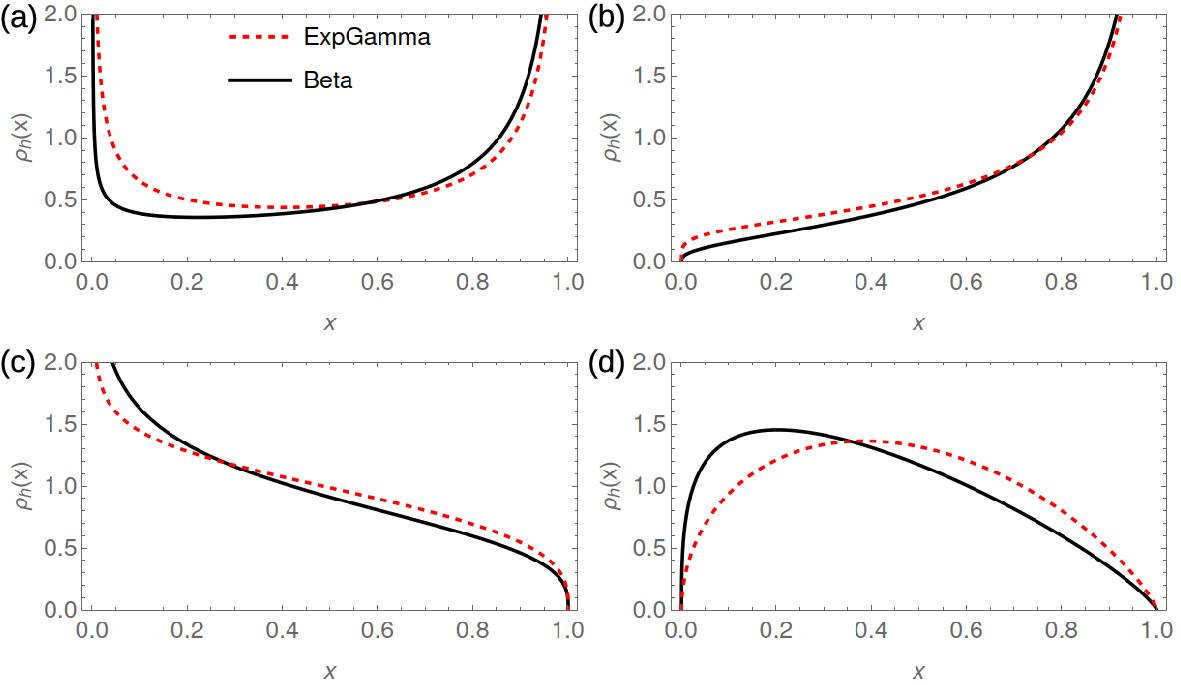}
\caption{Probability density function of the infectivity for model $C_1$ (dashed line, Eq.~\eqref{eq:rho_beta}) and the within-host microbial infection model proposed in Sec.~\ref{sec:Birth-Death_Particular} (continous line, Eq.~\eqref{eq:rhox_ExpGamma}). Different panels illustrate the dependence of $\rho_h^{(C_1)}(x;\alpha,\beta)$ on infectivity in each of the regimes described in the text: (a) Regime (i) with $(\alpha,\beta)=(0.5,0.25)$, (b) regime (ii) with $(\alpha,\beta)=(1.2,0.36)$, (c) regime (iii) with $(\alpha,\beta)=(0.89,1.32)$, and (d) regime (iv) with $(\alpha,\beta)=(1.5,1.8)$. When parametrised by the mean and variance of the effective microbe mortality considered in the microbial infection model of Sec.~\ref{sec:Birth-Death_Particular}, the four panels correspond to: (a) $(\mu_\lambda,v_\lambda)=(0.5,1.0)$, (b) $(\mu_\lambda,v_\lambda)=(0.3,0.25)$, (c) $(\mu_\lambda,v_\lambda)=(1.5,1.7)$ and (d) $(\mu_\lambda,v_\lambda)=(1.2,0.8)$.}
\label{fig_rhox_4regimes}
\end{figure} 

\begin{figure}
\centering
\includegraphics[width=10 cm]{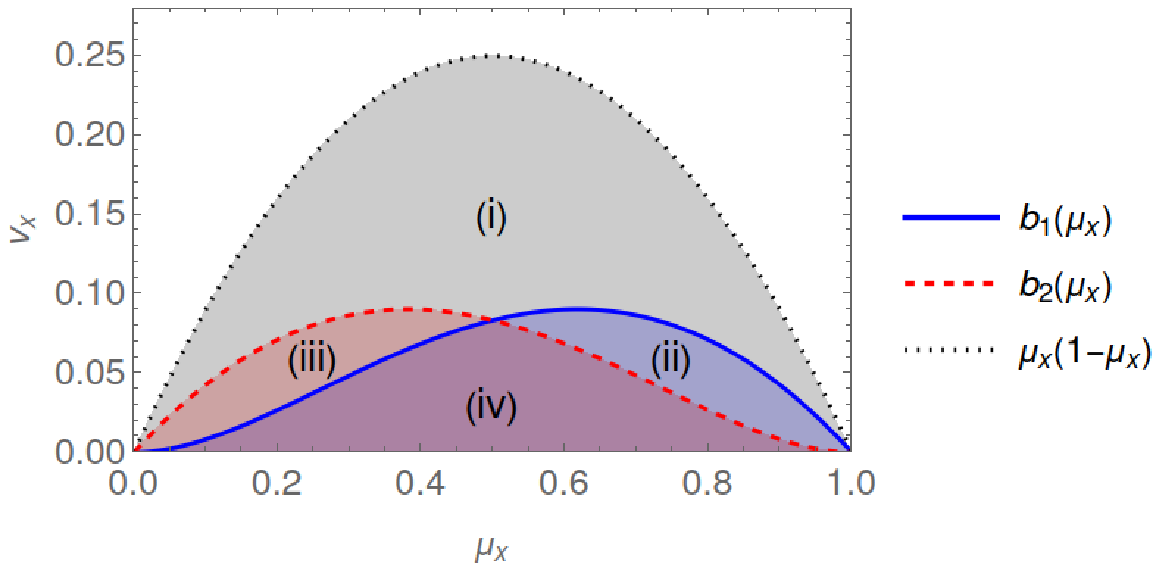}
\caption{Location of the four regimes for the shape of the infectivity PDF $\rho_h^{(C_1)}(x;\alpha,\beta)$ in the space spanned by the mean and variance of the infectivity.}
\label{fig_4Regimes_mux_vx}
\end{figure} 

A fully analytical study of the trends of the probability of infection $P^{(C_1)}$ when varying $\mu_x$ and $v_x$ is challenging and a graphical analysis will be presented instead. The continuous lines in Fig.~\ref{fig_PC1_vs_mu_and_v_2panels_neq50}(a) show the dependence of $P^{(C_1)}$ on the mean infectivity for two values of the variance. For given $v_x$, the condition \eqref{eq:vmax} implies that $P^{(C_1)}$ is only defined for $\mu_x \in (\mu_x^-, \mu_x^+)$, where 
\begin{equation}
\mu_x^{\pm}(v_x) = (1 \pm \sqrt{1-4v_x})/2~.
\end{equation}

At the boundaries $\mu_x^{\pm}$, the probability of infection satisfies
\begin{equation}
\label{eq:lim_PC1_mu+-}
    \lim_{\mu_x \rightarrow \mu_x^{\pm}} P^{(C_1)} = \mu_x^{\pm}~.
\end{equation}

Consequently, the endpoints of the curves for $P^{(C_1)}$ versus $\mu_x$ lie along a diagonal line with a slope of 1 in the $(\mu_x, P^{(C_1)})$ space (as indicated by the dotted line in Fig.~\ref{fig_PC1_vs_mu_and_v_2panels_neq50}(a)). When $\mu_x$ takes relatively small values, the probability $P^{(C_1)}$ increases with $\mu_x$ for any fixed value of $v_x$. This behavior aligns with our expectations based on Theorem \ref{th:GeneralExpectedPinf}, specifically Eq.~\eqref{eq:hatPApprox} with $v_n=0$. 

In contrast, for sufficiently high values of $\mu_x$, $P^{(C_1)}$ decreases. This might appear counterintuitive at first glance. However, this decrease is a consequence of keeping $v_x$ fixed. Indeed, when $\mu_x$ becomes large enough, the system enters regime (i) for any $v_x>0$ (refer to Fig.~\ref{fig_4Regimes_mux_vx}). In this regime, the divergence of $\rho_h^{(C_1)}(x;\alpha,\beta)$ becomes increasingly pronounced at both $x=0$ and $x=1$ as $\mu_x$ increases. The prominence of divergence at $x=0$ is essential for maintaining the fixed value of $v_x$. In more qualitative terms, as the mean infectivity increases, the likelihood of microbes having either high or low infectivity also rises.

The graphical representation of $P^{(C_1)}$ vs. $v_x$ suggests that $P^{(C_1)}$ monotonically decreases with increasing $v_x$ for any $n$ and $\mu_x$ (see the examples for two values of $\mu_x$ shown by the continuous curves in Fig.~\ref{fig_PC1_vs_mu_and_v_2panels_neq50}(b)). The general decrease of $P^{(C_1)}$ with $v_x$ predicted for small infectivities in Theorem \ref{th:GeneralExpectedPinf} seems to hold for any mean infectivity in this example. In the limit of small variance for the infectivity, 
\begin{equation}
\label{eq:limPC1_vx0}
    \lim_{v_x \searrow 0} P^{(C_1)} = 1-(1-\mu_x)^n~,
\end{equation}
i.e. model A is recovered. In the limit of high variance, one can use Euler's infinite product formula for $1/\Gamma(z)$ \cite{Abramowitz-Stegun} to show that 
\begin{equation}
\label{eq:limPC1_vxvmax}
    \lim_{v_x \nearrow v_{\max}(\mu_x)} P^{(C_1)} = \mu_x~.
\end{equation}

This is graphically illustrated by the examples in Fig.~\ref{fig_PC1_vs_mu_and_v_2panels_neq50}(b) which show that the curves of $P^{(C_1)}$ vs. $v_x$ end at the dotted line corresponding to $v_x \nearrow v_{\max}(\mu_x)$.

\begin{figure}
\centering
\includegraphics[width=12 cm]{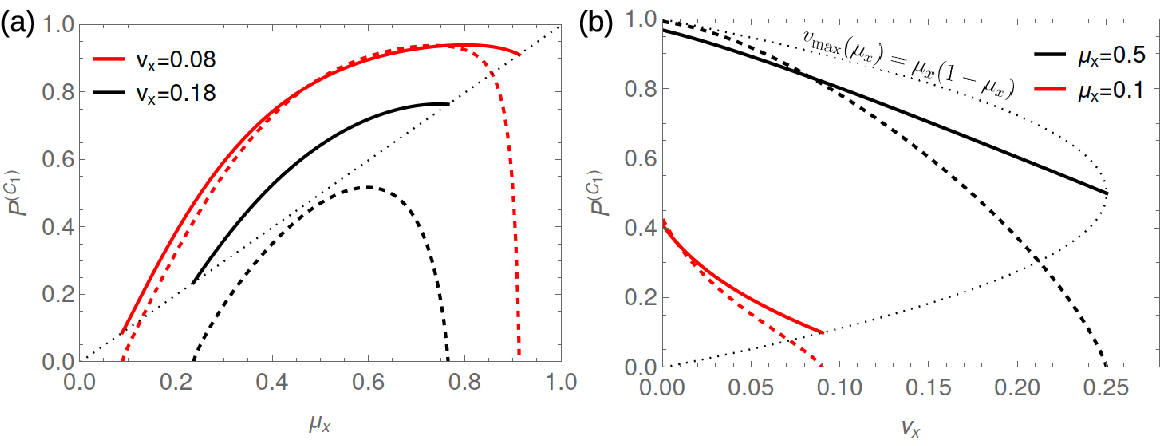}
\caption{Expected probability of infection for examples $C_1$ with dose $n=5$ (continuous lines, Eq.~\eqref{eq:Pinf_C1}) and $\tilde{C_1}$ (dashed lines, Eq.~\eqref{eq:Pinf_tildeC1}). Panel (a) shows the dependence of $P^{(C_1)}$ and $P^{(\tilde{C_1})}$ on the mean infectivity, $\mu_x$, for two values of the variance, $v_x$, as marked by the legend. The dotted diagonal line indicates the bounds for $P^{(C_1)}$  given by Eq.~\eqref{eq:lim_PC1_mu+-}.  Panel (b) shows the dependence of $P^{(C_1)}$ and $P^{(\tilde{C_1})}$ on the infectivity variance $v_x$. The dotted line shows the limit of $P^{(C_1)}$ as $v_x \nearrow v_{\max}(\mu_x)$ (cf. Eqs.~\eqref{eq:vmax} and \eqref{eq:limPC1_vxvmax}).}
\label{fig_PC1_vs_mu_and_v_2panels_neq50}
\end{figure} 

The validity of Corollary \ref{cor:n50_general} for model $C_1$ can be analytically studied for small and relatively homogeneous infectivity, as shown in Proposition \ref{pro:slope_C1}.

\begin{proposition}
\label{pro:slope_C1}
    If the infectivity distribution satisfies $\mu_x \ll 1$ and $v_x \ll \mu_x$, the slope of $P^{(C_1)}(\alpha,\beta,n)$ evaluated at the median dose decreases when increasing $v_x$. 
\end{proposition}

\begin{proof}
    From Eq.~\eqref{eq:alpha_beta_vs_mux_vx}, the conditions $\mu_x \ll 1$ and $v_x \ll \mu_x$ correspond to $\beta \gg \alpha$ and $\beta \gg 1$. Under these conditions,  Stirling's formula for the gamma function, $\Gamma(z) \sim e^{-z} z^{z-1/2} \sqrt{2 \pi}$, allows $P^{(C_1)}$ to be approximated as
    \begin{equation}
        \label{eq:Pinf_tildeC1}
        P^{(C_1)}(\alpha,\beta,n) \simeq P^{(\tilde{C_1})}(\alpha,\beta,n) \equiv 1-\left( 1+\frac{n}{\beta}\right)^{-\alpha}~.
    \end{equation}

The median infectious dose corresponding to $P^{(\tilde{C_1})}=1/2$ is $n_{50} = \beta (2^{1/\alpha}-1)$ and this yields the following slope for the infection curve:
\begin{equation}
    \left.\frac{\text{d} P^{(\tilde{C_1})}}{\text{d}n} \right|_{n=n_{50}} = \frac{\alpha}{\beta} 2^{1-1/\alpha}~.   
\end{equation}

Expressing the slope as a function of $\mu_x$ and $v_x$ through Eqs.~\eqref{eq:alpha_beta_vs_mux_vx} and taking the binary logarithm gives
\begin{equation}
    \log_2  \left.\frac{\text{d} P^{(\tilde{C_1})}}{\text{d}n} \right|_{n=n_{50}} = f^{(\tilde{C_1})}(\mu_x,v_x) \equiv 1+ \log_2 \frac{\mu_x}{1-\mu_x} -\frac{v_x}{\mu_x(\mu_x(1-\mu_x)-v_x)}~.
\end{equation}

Differentiating the logarithmic slope with respect to $v_x$ yields
\begin{equation}
    \frac{\partial f^{(\tilde{C_1})}(\mu_x,v_x)}{\partial v_x} = -\frac{(1-\mu_x)}{(v_x-\mu_x(1-\mu_x))^2} <0~,
\end{equation}
which proves the proposition.

\end{proof}

The approximate probability of infection $P^{(\tilde{C_1})}$ derived in the proof of Proposition \ref{pro:slope_C1} approaches zero at the boundaries $\mu_x^{\pm}(v_x)$ for the infectivity (see the dashed curves in Fig.~\eqref{fig_PC1_vs_mu_and_v_2panels_neq50}(a)). This significantly deviates from the behaviour of the exact probability $P^{(C_1)}$ given by Eq.~\eqref{fig_PC1_vs_mu_and_v_2panels_neq50}. This failure of the approximation is expected: On the one hand, the necessary condition $\mu_x \ll 1$ for $P^{(\tilde{C_1})}$ to approximate $P^{(C_1)}$ is not satisfied when $\mu_x \nearrow \mu_x^{+}$. On the other hand, the condition $v_x \ll \mu_x$ is violated when $\mu_x \searrow \mu_x^{-}$. 

Comparing the continuous and dashed lines in Fig.~\eqref{fig_PC1_vs_mu_and_v_2panels_neq50}(b) reveals that $P^{(\tilde{C_1})}$ offers a better approximation to $P^{(C_1)}$ for small values of $v_x$ and $\mu_x$ (compare the good agreement between the curves for small $v_x$ with the poorer agreement between the black curves for small $v_x$). In the limit of vanishing variance, the approximate probability approaches the value
\begin{equation}
    \label{eq:limPC2_vx0}
    \lim_{v_x \searrow 0} P^{(\tilde{C_1})} = 1-e^{-n\mu_x/(1-\mu_x)}~.
\end{equation}
This approximate value exceeds the exact probability $P^{(C_1)}$ (see Eq.~\eqref{eq:limPC1_vx0}) since $e^{-\mu_x/(1-\mu_x)}>1-\mu_x$ for any $\mu_x \in (0,1)$. 

Fig.~\eqref{fig_PC1_vs_mu_and_v_2panels_neq50}(b) also shows that the deviation of $P^{(\tilde{C_1})}$ from $P^{(C_1)}$ worsens as $v_x$ increases. For given $\mu_x$, this deviation is more pronounced when $v_x$ approaches the maximal variance, $v_{\max}(\mu_x)$. In such cases, $P^{(\tilde{C_1})}$ approaches zero and significantly deviates from the corresponding limit for $P^{(C_1)}$ given by Eq.~\eqref{eq:limPC1_vxvmax}.

\subsection{Models of type $C^\prime$: Between-host heterogeneous infectivity and dose}
\label{subsec:ModelCprime}

Models of type $C^{\prime}$ incorporate the three types of heterogeneity described in Fig.~\ref{fig_Models_space_3D}. These models can either be regarded as extensions of models of type $A^\prime$ to incorporate heterogeneities of type I and II or extensions of models of type $C$ to account for type III heterogeneity. Following this, the expected probability of infection for models of type $C^{\prime}$ can be calculated in several equivalent ways:
\begin{align}
\label{eq:PCprime_0}
    P^{(C^{\prime})}(\bm\xi,\bm\xi_{\text{C}}) =& \sum_{n =1}^{\infty} p^{\text{d}}(n;\bm\xi) \int_0^1 P^{(A)}(x,n)  \rho_h(x;\bm\xi_{\text{C}}) \text{d}x~,\\
    \label{eq:PCprime_1}
    =& \sum_{n =1}^{\infty} P^{(C)}(n,\bm\xi_{\text{C}}) p^{\text{d}}(n;\bm\xi)~,\\
    \label{eq:PCprime_2}
    =& \int_0^1 P^{(A')}(x,\bm\xi) \rho_h(x;\bm\xi_{\text{C}}) \text{d}x~.
\end{align}


The general results for the expected probability of infection in the limit of small infectivity presented in Sec.~\ref{subsec:GeneralExpectedPinf} directly apply to models of type $C^{\prime}$.

The next subsections present two models of type $C^\prime$ that have been proposed in previous works. Both assume beta-distributed infectivity (i.e., PDF given by Eq.~\eqref{eq:rho_beta}). As a consequence, both examples will exhibit the four regimes described in Sec.~\ref{sec:C1}.

\subsubsection{Example $C_1^\prime$: Beta-distributed infectivity and Poisson-distributed dose}
\label{sec:Example_C1prime}

Here, it is assumed that the dose is $n_h \sim \text{Pois}(\mu_n)$ so that model $C_1^\prime$ can be interpreted as an extension of model $A_1^\prime$ (Sec.~\ref{subsec:A1prime}) to include a beta-distributed infectivity. Using Eq.~\eqref{eq:PAprime1} to set $P^{(A')}(x,\bm\xi)=P^{(A^{\prime}_1)}(x,\mu_n)$ in Eq.~\eqref{eq:PCprime_2} leads to the following expected probability of infection:

\begin{equation}
\label{eq:P_C1p}
    P^{(C_1^{\prime})}(\alpha,\beta,\mu_n) =  1 - ~_1F_1 (\alpha,\alpha+\beta,-\mu_n)~.
\end{equation}

Here, $_1F_1$ is the Kummer confluent function~\cite{Abramowitz-Stegun}. Eq.~\eqref{eq:P_C1p} is the exact beta-Poisson model proposed by  Furumoto and Mickey  \cite{Furumoto1967}. In the same work, the authors derived the beta-Poisson approximation,
\begin{equation}
\label{eq:Pinf_tildeC1prime}
    P^{(C_1^{\prime})}(\alpha,\beta,\mu_n) \simeq P^{(\tilde{C_1}^\prime)}(\alpha,\beta,\mu_n) \equiv 1-\left( 1+\frac{\mu_n}{\beta}\right)^{-\alpha}~,
\end{equation}
which is valid for $\beta \gg \alpha$ and $\beta \gg 1$, i.e. for $\mu_x \ll 1$ and $v_x \ll \mu_x$ (cf. Eq.~\eqref{eq:alpha_beta_vs_mux_vx}). 

The functional relationship between $P^{(\tilde{C_1}^\prime)}$ and $\mu_n$ mirrors that of $P^{(\tilde{C_1})}$ and $n$ (compare Eqs.~\eqref{eq:Pinf_tildeC1prime} and \eqref{eq:Pinf_tildeC1}). Consequently, Proposition \ref{pro:slope_C1}, established for model $C_1$, directly extends to model $C_1^{\prime}$ through the expression of the median dose in terms of $\mu_n$. While this proposition may bear resemblance to previous demonstrations expressing $P^{(\tilde{C_1}^\prime)}$ as a function of $\alpha$ and the median dose~\cite{Haas_Book2014}, it is essential to highlight a crucial distinction. Proposition~\ref{pro:slope_C1} elucidates the decline in the slope of $P^{(\tilde{C_1}^\prime)}$ with increasing $v_x$ while holding $\mu_x$ constant. Such delineation between the variance and mean of infectivity cannot be achieved by parameterising $P^{(\tilde{C_1}^\prime)}$ in terms of $\alpha$ and the median dose.

Despite the formal similarity between Eqs.~ \eqref{eq:Pinf_tildeC1} and \eqref{eq:Pinf_tildeC1prime} and the fact that both are valid in the limit of low and relatively homogeneous infectivity, $P^{(\tilde{C_1}^\prime)}$ does not result by using $P^{(\tilde{C_1})}$ and the Poisson probability mass function in Eq.~\eqref{eq:PCprime_1}.  Irrespective of that, the deviations of $P^{(\tilde{C_1}^\prime)}$ from the exact probability $P^{(C_1^\prime)}$ are analogous to those described in Sec.~\ref{sec:C1} for $P^{(\tilde{C_1})}$ compared to $P^{(C_1)}$.

\subsubsection{Example $C_2^\prime$: Beta-distributed infectivity and negative binomial dose distribution}

This model assumes a negative binomial distribution for the dose, $n_h \sim \text{NB}(r,p)$. It can be viewed as an extension of model $A_2^\prime$ (Sec.~\ref{subsec:A2prime}) to incorporate a beta-distributed infectivity. Setting $P^{(A')}(x,\bm\xi)=P^{(A^{\prime}_2)}(x,r,p)$ in Eq.~\eqref{eq:PCprime_2} gives the probability of infection~\cite{Haas_RiskAnalysis2002},
\begin{equation}
    P^{(C_2^{\prime})}(\alpha,\beta,r,p) =  1 -  ~_2F_1 (r,\beta,\alpha+\beta,1-p) p^r~,
\end{equation}
where $_2F_1$ is the Gaussian hypergeometric function~\cite{Abramowitz-Stegun}.

The following Corollaries about the dependence of model $C_2^\prime$ on the variability of the dose follow from the results presented in Sec.~\ref{subsec:A2prime} for model $A_2^\prime$.

\begin{corollary}
\label{cor:Cprime2_reduces_to_Cprime1}
    The expected infection probability of model $C_2^{\prime}$ reduces to that of model $C_1^{\prime}$ when the variance of the dose approaches the mean from above (i.e. in the limit $v_n \searrow \mu_n$).
\end{corollary}
\begin{proof}
    Expressing the expected probability of infection for model $C_2^\prime$ as
    \begin{equation}
    \label{eq:C2prime_2}
        P^{(C_2^{\prime})}(\alpha,\beta,r,p) = \int_0^1 P^{(A_2')}(x,r,p) \rho^{(C_1)}_h(x;\alpha,\beta)  \text{d}x~,
    \end{equation}
    the corollary easily follows from Proposition \ref{pro:Aprime2_reduces_to_Aprime1}.
\end{proof}

\begin{corollary}
    \label{pro:PCprime2_variance}
    For a fixed mean dose, the expected probability of infection of model $C_2^{\prime}$ is a decreasing function of the dose variance.
\end{corollary}

\begin{proof}
    This is a straightforward consequence of Proposition \ref{pro:PAprime2_variance} when $P^{(C_2^{\prime})}(\alpha,\beta,r,p)$ is expressed as in Eq.~\eqref{eq:C2prime_2}.
\end{proof}

The following result establishes a hierarchy for  the expected probability of infection predicted by model $A$ and all models with heterogeneous infectivity of types $C$ and $C^\prime$. This highlights the fact that increasing heterogeneity in infectivity and dose leads to a decrease of the expected probability of infection.

\begin{proposition}
\label{pro:Inequalities_PA_PC_PCp}
    The following inequalities hold for any infectivity and dose distribution parameters $\mu_x$, $v_x$, $\mu_n$ and $v_n$:
    \begin{equation}
    \label{eq:Inequalities_PA_PC_PCp}
        P^{(A)}(\mu_x,\mu_n) \geq P^{(C_1)}(\mu_x,v_x,\mu_n) \geq
        P^{(C_1^\prime)}(\mu_x,v_x,\mu_n) \geq P^{(C_2^\prime)}(\mu_x,v_x,\mu_n,v_n)~.
    \end{equation} 
\end{proposition}

\begin{proof}
    The first inequality arises from the non-convex nature of $P^{(A)}(x,\mu_n)$ with respect to $x$. Jensen's inequality implies that $P^{(A)}(\E_x(x),\mu_n) \geq \E_x(P^{(A)}(x,\mu_n))$. This equivalently yields $P^{(A)}(\mu_x,\mu_n) \geq P^{(C_1)}(\mu_x,v_x,\mu_n)$ since $\mu_x=\E_x(x)$ and $P^{(C_1)}(\mu_x,v_x,\mu_n)$  can be expressed as $\E_x(P^{(A)}(x,\mu_n))$ through Eq.~\eqref{eq:Pinv_C}.

The second inequality stems from the non-convexity of $P^{(C_1)}(\mu_x,v_x,n)$ with respect to $n$. Applying Jensen's inequality, leads to $P^{(C_1)}(\mu_x,v_x,\E_n(n)) \geq \E_n(P^{(C_1)}(\mu_x,v_x,n))$. This reduces to $P^{(C_1)}(\mu_x,v_x,\mu_n) \geq
    P^{(C_1^\prime)}(\mu_x,v_x,\mu_n)$ using $\mu_n=\E_n(n)$ and $P^{(C_1^\prime)}(\mu_x,v_x,\mu_n) = \E_n(P^{(C_1)}(\mu_x,v_x,n))$ (cf. Eq.~\eqref{eq:PCprime_1}).

    The third inequality follows from the Corollaries \ref{cor:Cprime2_reduces_to_Cprime1} and \ref{pro:PCprime2_variance}.

    Fig.~\ref{P_ModelsA_C1_Cp1_Cp2} illustrates the inequalities given in Eq.~\eqref{eq:Inequalities_PA_PC_PCp} for fixed $v_x$, $\mu_n$ and $v_n$ and variable $\mu_x$.
    
\end{proof}

\begin{figure}
\centering
\includegraphics[width=8 cm]{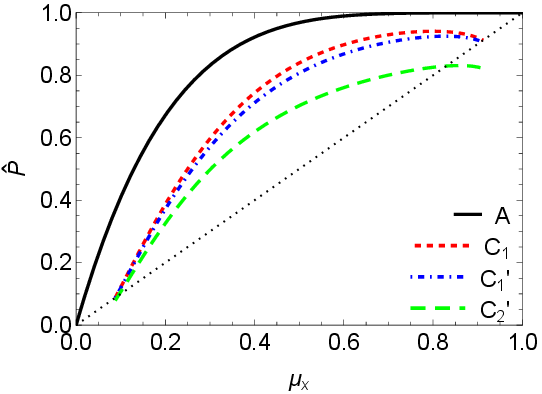}
\caption{Graphical illustration of the inequalities given in Proposition \ref{pro:Inequalities_PA_PC_PCp}. The dependence of the probability of infection on the mean infectivity for a system with homogeneous infectivity (Model $A$, continuous line) is compared to the dependence on $\mu_x$ of the expected probability of infection of models $C_1$ (dashed line), $C^{\prime}_1$ (dot-dashed line) and $C^{\prime}_2$ (long dashed line). The dose was set to $n=5$ for model $A$. The mean and variance of the dose were set to $\mu_n=5$ and $v_n=20$, respectively, for models $C_1$, $C^{\prime}_1$ and $C^{\prime}_2$.}
\label{P_ModelsA_C1_Cp1_Cp2}
\end{figure}

\section{A model for microbial infection with random infectivity}
\label{sec:GrowthModel}
This section provides a biologically explicit interpretation of single-hit dose-response models of type $C$, which account for variations in pathogen infectivity across different hosts. Between-host heterogeneity in infectivity is assumed to arise from biological factors that influence the ability of microbes to establish an infection. In line with Assumption 1 in the introduction, the proposed framework allows for infection to be initiated by a single microbe but specifically considers cases where microbial growth is essential for infection establishment.

Inspired by compartmental models of infection dynamics within a host, it is assumed that infection occurs if the microbial reproduction number within the host, $R$, satisfies $R>1$ \cite{DeLeenheer2003,Li2020,Elbaz2023}. Conversely, if $R<1$, the infection fails to establish. For instance, in a minimal within-host model, microbial dynamics can be described by a simple growth mechanism where microbes infect cells at rate $\beta$ and die at rate $\delta$ ~\cite{Nowak_Book2006,Wodarz2007}. At the early stages of infection, the microbial population $N$ is assumed to obey the equation $\text{d}N/\text{d}t=(\beta-\delta)N$. A necessary condition for infection to occur is $R=\beta/\delta>1$, meaning each microbe produces more than one offspring on average during its average lifetime, $1/\delta$. In contrast, a microbial population dies out when $R<1$. While more complex within-host models incorporate additional parameters, the fundamental infection criterion $R>1$ remains broadly applicable \cite{DeLeenheer2003,Li2020,Elbaz2023}.

\subsubsection{General formulation of the model}

Between-host heterogeneity is modelled by assuming that the basic reproductive ratio $R_{i,h}$ of a microbe $i$ in host $h$ is a random variable with PDF $\rho_R(R_{i,h};\bm\xi_h)$, where $\bm\xi_h$ is a set of host-dependent parameters. The infectivity of a randomly chosen microbe within host $h$ is the probability that  $R_{i,h}>1$:
\begin{equation}
\label{eq:xh_f}
    x_h = f(\bm\xi_h) \equiv \int_1^{\infty} \rho_R(R;\bm\xi_h) \text{d} R~.
\end{equation}

The variability of infectivity between hosts is described by assuming that the parameters $\bm\xi_h$ are random variables with joint PDF $\rho_{\xi}(\bm\xi_h; \bm\xi^\prime)$. Here, $\bm\xi^\prime$ are parameters for the PDF of $\bm\xi_h$. Under this assumption, the probability density function for the infectivity within a randomly chosen host is given by \cite{Riley-Hobson-Bence_MathsBook_3rdEd}:
\begin{equation}
\label{eq:rhox_general}
    \rho_h(x;\bm\xi^\prime)=\int_{\text{supp}(\bm\xi_h)} \text{d}\bm\xi_h \rho_{\xi}(\bm\xi_h; \bm\xi^\prime) \delta(f(\bm\xi_h)-x)~,
\end{equation}
where $\text{supp}(\bm\xi_h)$ is the support of the distribution of $\bm\xi_h$ and  $\delta(\cdot)$ is the Dirac delta function.

In the particular case in which the PDF for $R$ involves a single parameter $\lambda$ (i.e. $\bm\xi_h = \lambda$), Eq.~\eqref{eq:rhox_general} reduces to

\begin{equation}
\label{eq:rhox_general_1D}
 \rho_h(x;\bm\xi^\prime)=\left| \frac{\text{d} f^{-1}(x)}{\text{d} x }\right| \rho_\xi(\lambda=f^{-1}(x);\bm\xi^\prime)~.
\end{equation}

\subsubsection{Example: Linking heterogeneous infectivity and microbial birth and death}
\label{sec:Birth-Death_Particular}

If the basic reproductive ratio is exponentially distributed with rate parameter $\lambda$, Eq.~\eqref{eq:xh_f} yields an infectivity
\begin{equation}
\label{eq:xh_f_particular}
    x = \int_1^{\infty} \lambda e^{-\lambda R} \text{d}R = e^{-\lambda} \equiv f(\lambda)~.
\end{equation}

The parameter $\lambda$  can be interpreted as an effective microbial death rate since the chances for a microbe to grow decrease for increasing $\lambda$.

Since $\lambda>0$, it is reasonable to assume that its value is drawn from a gamma distribution with PDF 
\begin{equation}
\label{eq:rho_xi_particular}
    \rho_{\xi}(\bm\xi_h; \bm\xi^\prime) = \rho_{\xi}(\lambda;\hat{\alpha},\hat{\beta}) = \frac{\hat{\alpha}^{\hat{\beta}}}{\Gamma(\hat{\beta})} \lambda^{\hat{\beta}-1} 
e^{-\hat{\beta} \lambda}~,
\end{equation}
where $\hat{\alpha}$ and $\hat{\beta}$ are the rate and shape parameters, respectively. 

Introducing now Eqs.~\eqref{eq:xh_f_particular} and \eqref{eq:rho_xi_particular} into Eq.~\eqref{eq:rho_xi_particular} leads to the following PDF for the infectivity:
\begin{equation}
\label{eq:rhox_ExpGamma}
    \rho_h(x;\hat{\alpha},\hat{\beta})=\frac{\hat{\alpha}^{\hat{\beta}}}{\Gamma(\hat{\beta}) } (-\ln x)^{\hat{\beta}-1} x^{\hat{\alpha}-1}~.
\end{equation}

This PDF does not seem to correspond to any well-known probability distribution. It will be referred to as the \emph{exponential gamma distribution} (denoted as $x \sim \text{ExpGamma}(\hat{\alpha},\hat{\beta})$) since it corresponds to a random variable $x$ whose exponential is gamma distributed. 

The influence on infectivity of randomness in the effective microbe mortality $\lambda$ can be elucidated by expressing the parameters $\hat{\alpha}$ and $\hat{\beta}$ in terms of the mean and variance of the effective death rate, $\mu_\lambda$ and $v_\lambda$, as follows: 
\begin{align}
\label{eq:s_vs_statskappa}
\hat{\alpha}&=\frac{\mu_{\lambda}}{v_{\lambda}}~\\
\label{eq:theta_vs_statskappa}
\hat{\beta}&=\frac{\mu_{\lambda}^2}{v_{\lambda}}~.
\end{align}

These identities can also be used as approximate relations between the moments of the effective microbe mortality and the parameters of the beta distribution used in models $C$, $C_1^\prime$ and $C_2^\prime$. Indeed, the $\text{ExpGamma}(\hat{\alpha},\hat{\beta})$ distribution reduces to the $\text{Beta}(\alpha,\beta)$ distribution with $\alpha=\hat{\alpha}$ and $\beta=\hat{\beta}$ for $x$ close to 1. This directly follows from the comparison of Eqs.~\eqref{eq:rho_beta} and \eqref{eq:rhox_ExpGamma} for $x \nearrow 1$ which yields $-\ln x \simeq (1-x)$. Fig.~\ref{fig_rhox_4regimes} shows the similarity of the PDFs $\rho_h(x;\hat{\alpha},\hat{\beta})$ and $\rho_h^{(C_1)}(x;\alpha,\beta)$ (Eq.~\eqref{eq:rho_beta}) for examples with $\alpha=\hat{\alpha}$ and $\beta=\hat{\beta}$.

Eq.~\eqref{eq:s_vs_statskappa} shows that $\hat{\alpha}$ is the inverse of the relative variance of $\lambda$ ($\text{RV}_\lambda = v_\lambda/\mu_\lambda$), meaning that it decreases as the microbe mortality becomes more heterogeneous between hosts. This approximately applies to the coefficient $\alpha$ of the beta distribution. 

Eq.~\eqref{eq:theta_vs_statskappa} shows that $\hat{\beta}$ is the square of the inverse of the coefficient of variation of $\lambda$ ($\text{CV}_\lambda = v_\lambda^{1/2}/\mu_\lambda$). Accordingly, $\hat{\beta}$ (and approximately $\beta$) also decreases with increasing heterogeneity in $\lambda$.

Fig.~\ref{fig_rhox_4regimes} graphically demonstrates that the dependence of $\rho_h(x;\hat{\alpha},\hat{\beta})$ on $x$ falls into four different regimes which are analogous to those defined for $\rho_h^{(C_1)}(x;\alpha,\beta)$ in Sec.~\ref{sec:C1}. The regions of each regime in the space $(\mu_\lambda,v_\lambda)$ are defined as follows:

\begin{description}
 \item[Regime (i) - Highly heterogeneous infectivity:] It is observed for any mean effective mortality rate $\mu_\lambda>0$ provided the effective mortality variance is high enough, $v_\lambda > \max \{\mu_\lambda,\mu_\lambda^2 \}$ (see the boundaries marked with continuous and dotted lines in Fig.~\ref{fig_mux_vx_vs_mulambda_vlambda}).
 \item[Regime (ii) - High infectivity with intermediate heterogeneity:] This corresponds to low microbe mortality, $\mu_\lambda <1$, and intermediate effective mortality variance, $\mu_\lambda^2 < v_\lambda < \mu_\lambda$.
 \item[Regime (iii) - Low infectivity with intermediate heterogeneity:] It corresponds to situations with relatively high effective mortality, $\mu_\lambda > 1$ and intermediate variance, $\mu_\lambda < v_\lambda < \mu_\lambda^2$.
 \item[Regime (iv) - Low heterogeneity:] This regime is observed for any $\mu_\lambda>0$, provided the effective mortality variance is below a certain level, i.e. for $v_\lambda < \min \{\mu_\lambda,\mu_\lambda^2 \}$.
\end{description}

\begin{figure}
\centering
\includegraphics[width=14 cm]{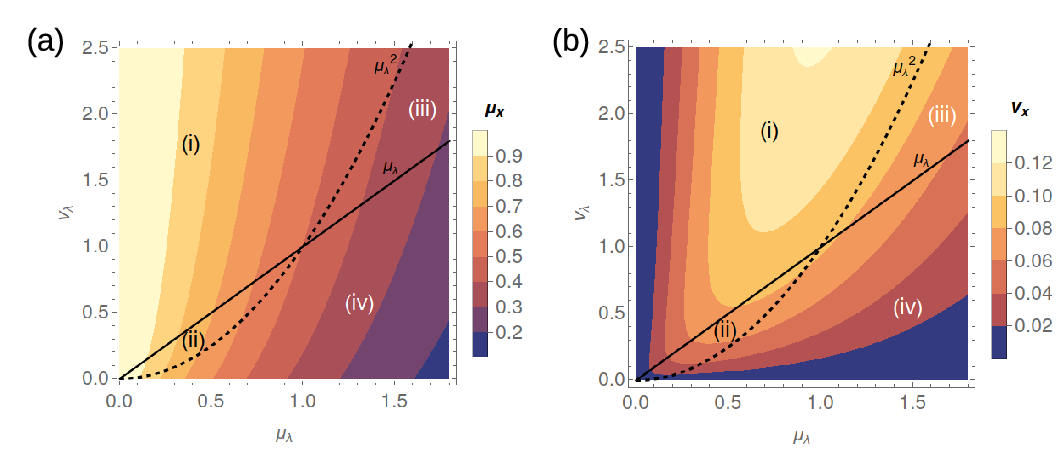}
\caption{Contour plots depict the interdependencies between (a) the mean ($\mu_x$) and (b) the variance ($v_x$) of infectivity, with the mean ($\mu_\lambda$) and variance ($v_\lambda$) of microbial mortality in the growth model proposed in Sec.~\ref{sec:Birth-Death_Particular}. The four regimes for the ExpGamma infectivity PDF ($\rho_h(x;\hat{\alpha},\hat{\beta})$, Eq.~\eqref{eq:rhox_ExpGamma}) are indicated in both panels, together with the regime boundaries corresponding to $v_\lambda = \mu_\lambda$ (continuous line) and $v_\lambda = \mu_\lambda^2$ (dotted line).}
\label{fig_mux_vx_vs_mulambda_vlambda}
\end{figure}

The expectation and variance of the infectivity can be calculated from Eq.~\eqref{eq:rho_xi_particular} and expressed as a function of $\mu_\lambda$ and $v_\lambda$ as follows:
\begin{align}
\label{eq:Ex_LogPower}
\mu_x&=(\frac{\hat{\alpha}}{1+\hat{\alpha}})^{\hat{\beta}}=\left(\frac{\mu_\lambda}{\mu_\lambda+v_\lambda}\right)^{\mu_\lambda^2/v_\lambda}~,\\
\label{eq:Varx_LogPower}
v_x&=(\frac{\hat{\alpha}}{2+\hat{\alpha}})^{\hat{\beta}}-(\frac{\hat{\alpha}}{1+\hat{\alpha}})^{2 \hat{\beta}} =\left(\frac{\mu_\lambda}{\mu_\lambda
+2 v_\lambda}\right)^{\mu_\lambda^2/v_\lambda}- \left(\frac{\mu_\lambda}{\mu_\lambda+v_\lambda} \right)^{2\mu_\lambda^2/v_\lambda}~.
\end{align}

As expected, the mean infectivity of the microbe, $\mu_x$, decreases with $\mu_\lambda$ (see the contour levels in Fig.~\ref{fig_mux_vx_vs_mulambda_vlambda}(a)), i.e. higher mortality of the microbe leads to lower infectivity. Conversely, there is a positive association between $\mu_x$ and $v_\lambda$, indicating that greater variability in microbe mortality tends to enhance infectivity on average. This enhanced infectivity in populations with heterogeneous mortality rates is likely driven by microbes with lower-than-average $\lambda$, which survive longer and are more likely to infect the host than microbes in a homogeneous population with the same mean mortality $\mu_\lambda$.

Regions characterised by low infectivity variance, $v_x$, are evident in regimes (i) and (vi), visible as dark areas in Fig.~\ref{fig_mux_vx_vs_mulambda_vlambda}(b). In these contexts, models featuring homogeneous infectivity (such as type $A$ and $A^\prime$) can offer a reasonable description of the expected probability of infection. In regime (i), the low infectivity variance region corresponds to situations with high mean infectivity, which forces the maximum infectivity variance $v_{\max}(\mu_x)$ to be small (see Eq.~\eqref{eq:vmax}). Conversily, regime (iv) corresponds to systems characterised by relatively homogeneous infectivity, irrespectively of the mean infectivity.

Elevated values of $v_x$ manifest along a ridge situated within regime (i), where neither $\mu_\lambda$ nor $v_\lambda$ are negligible. This observation implies that systems characterised by both substantial expectation and variance in effective mortality are more effectively characterised by models of type $C$ and $C^\prime$, which incorporate variability in infectivity.

\section{Conclusion}
\label{sec:conclusion}
This work has delved into the impact of variation in dose and microbial infectivity on infection probability within the framework of single-hit dose-response models. It has been rigorously shown that heterogeneity in the infectivity of microbial populations ingested by a host enhances the chances of infection (heterogeneity of type I, Theorem~\ref{th:Within-Host-Inf}). This finding should be differentiated from the observation that the expected probability of infection solely depends on the mean infectivity, $\mu_m$ (Sec.~\ref{subsec:ModelB}). In essence, the expected probability represents the average infection likelihood across a group of hosts exposed to doses with heterogeneous infectivity, rather than the infection probability of an individual.

The presence of heterogeneity of type II, where microbial infectivity varies across hosts, has been demonstrated to reduce the expected probability of infection. General results have been given for microbial populations with small infectivity (Theorem~\ref{th:GeneralExpectedPinf}). Moreover, the ubiquitous flattening of the dose-response curve observed in experiments has also been shown to be associated with an increase in the variance of infectivity. The validity of these results has been illustrated with several models which assume random infectivity between hosts. Across these models, the expected probability of infection decreases when increasing the heterogeneity of infectivity regardless of the infectivity expectation, as summarised in Proposition~\ref{pro:Inequalities_PA_PC_PCp} (cf. $P^{(A)}$ and $P^{(C_1)}$).   This suggests that the small infectivity limit required in Theorem~\ref{th:GeneralExpectedPinf} may not be overly restrictive. Investigating a more general theorem which relaxes this condition might be interesting. 

Heterogeneity of type III, associated with variations in dose size among hosts, has similarly been demonstrated to reduce the expected probability of infection.  Theorem~\ref{th:GeneralExpectedPinf} establishes this result in the limit of small infectivity. Specific examples of models show that an increase of the expected probability of infection with the variance of the dose holds for any mean infectivity (summarised in Proposition~\ref{pro:Inequalities_PA}). Once again, this indicates a mild role played by the required small infectivity to probe Theorem~\ref{th:GeneralExpectedPinf}. 

The last results section has proposed a within-host microbial growth model with randomly distributed reproductive ratio within a host. An example of this model has been presented, utilising the effective microbial death rate as a parameter, which essentially acts as the inverse of the reproductive ratio. The model provides interesting predictions: near-uniform infectivity is anticipated when either the mean or variance of the microbial death rate is small. Conversely, systems characterised by substantial mean and variance in the effective death rate will exhibit heterogeneous infectivity. Furthermore, this model offers a biological interpretation for the parameters of the beta-Poisson dose-response model (referred to as Example $C_1^\prime$ within the framework presented here). Specifically, it has been found that both $\alpha$ and $\beta$ decrease with increasing heterogeneity in the effective microbial death rate. In contrast, both $\alpha$ and $\beta$ increase when increasing the expectation of the effective microbial death rate.

Exploring experimental validations to assess the anticipated impacts of diverse forms of heterogeneity is of significant scientific interest, but it also presents considerable challenges. For instance, consider the prediction of Theorem~\ref{th:Within-Host-Inf}, which suggests that the infection probability of a host is minimal for doses with homogeneous infectivity. Strictly speaking, testing this would require a counterfactual approach~\cite{Lash-Rothman_ModernEpidemiology_4thEd}, where the probability of a host being infected by a dose with homogeneous infectivity is compared with the probability of the same host being infected by heterogeneous doses with the same mean infectivity. However, ensuring consistent hosts across different experiments poses challenges, as hosts may change over time or differ between tests. Using very similar hosts might mitigate this issue.

Even if confounding associated with host differences is addressed, significant challenges remain regarding microbial heterogeneity and dose variability. Accurately quantifying and controlling the infectivity of microbial populations is inherently difficult. Furthermore, variations in dose delivery methods, host responses, and environmental factors can introduce additional variability, complicating the interpretation of results.

\section*{Acknowledgements}
The author acknowledges fruitful discussions with Norval Strachan,  Ovidiu Rotariu and Ken Forbes as well as funding support from a Medical Research Council Fellowship (MR/W021455/1).








\end{document}